\documentclass[11pt,english]{article}
\usepackage[LGR,T1]{fontenc}
\usepackage[latin9]{inputenc}
\usepackage{amsmath}
\usepackage{amsthm}
\usepackage{amssymb}

\makeatletter

\newcommand*\LyXPunctSpace{\hphantom{,}}
\DeclareRobustCommand{\greektext}{%
  \fontencoding{LGR}\selectfont\def\encodingdefault{LGR}}
\DeclareRobustCommand{\textgreek}[1]{\leavevmode{\greektext #1}}
\ProvideTextCommand{\~}{LGR}[1]{\char126#1}

\theoremstyle{plain}
\newtheorem{thm}{\protect\theoremname}
  \theoremstyle{plain}
  \newtheorem{prop}[thm]{\protect\propositionname}
  \theoremstyle{definition}
  \newtheorem*{example*}{\protect\examplename}
  \newtheorem{corollary}{\protect\corollaryname}

\usepackage{mathrsfs}\usepackage{amsthm}\usepackage{amscd}

\usepackage{hyperref}
\usepackage{color,graphicx,stmaryrd}

\usepackage{a4wide}
\definecolor{colorlinks}{RGB}{0, 24, 168}
\definecolor{colorcites}{RGB}{124, 10, 2}
\hypersetup{
	colorlinks=true,
	linkcolor=colorlinks,
	citecolor=colorcites,
	urlcolor=colorlinks,
	pdfborder={0 0 0}
}

\newcounter{EQNR}
\renewcommand{\contentsname}{Contents\\
{\footnotesize\normalfont(A table of contents should normally not
be included)}}

\makeatother

\usepackage{babel}
  \providecommand{\examplename}{Example}
  \providecommand{\propositionname}{Proposition}
\providecommand{\theoremname}{Theorem}
\providecommand{\corollaryname}{Corollary}

\begin{document}

\title{The discrete analogue of the Gaussian}

\author{Gautam Chinta \and Jay Jorgenson \footnote{The first and second-named authors acknowledge grant support
from PSC-CUNY Awards 65400-00-53 and 65400-00-55, which are jointly funded
by the Professional Staff Congress and The City University of New York.} \and Anders Karlsson
\footnote{The third-named author acknowledges grant support by the Swiss NSF grants 200020-200400, 200021-212864 and the Swedish
Research Council grant 104651320.} \and Lejla Smajlovi\'{c}}

\date{September 16, 2024}
\maketitle

\begin{abstract}\noindent
This paper illustrates the utility of the heat kernel on $\mathbb{Z}$ as the discrete analogue of the Gaussian density function. It is the two-variable function $K_{\mathbb{Z}}(t,x)=e^{-2t}I_{x}(2t)$ involving a Bessel function and variables $x\in\mathbb{Z}$ and real $t\geq 0$. Like its classic counterpart it appears in many mathematical and physical contexts and has a wealth of applications. Some of these will be reviewed here, concerning Bessel integrals, trigonometric sums, hypergeometric functions and asymptotics of discrete models appearing in statistical and quantum physics. Moreover, we prove a new local limit theorem for sums of integer-valued random variables, obtain novel special values of the spectral zeta function of Bethe lattices, and provide a discussion on how $e^{-2t}I_{x}(2t)$ could be useful in differential privacy.
\end{abstract}

\section{Introduction}

The \emph{Gaussian density function}, or simply \emph{Gaussian}, is the two-variable function defined by
\begin{equation}\label{eq:real_heat}
K_{\mathbb{R}}(t,x):=\frac{1}{\sqrt{4\pi t}}e^{-x^{2}/4t}
\,\,\,\,\,
\text{\rm for}
\,\,\,\,\,x \in \mathbf{R}
\,\,\,\,\,
\text{\rm and}
\,\,\,\,\,t \in \mathbf{R}^{+}.
\end{equation}
In common language the Gaussian is sometimes referred to as the bell-shaped curve. %\emph{bell-shaped curve}.
In probability theory, \eqref{eq:real_heat} is the density function of a normal random variable. %\emph{normal random variable}.
The universality of the Gaussian is evident from the  classical central limit theorem,
the theory of Brownian motion, and that \eqref{eq:real_heat} serves as the fundamental
solution to the heat equation on the real line
\begin{equation}\label{eq:real_heateqn}
\left(-\frac{\partial^{2}}{\partial x^{2}}+\frac{\partial}{\partial t}\right)f(t,x)=0.
\end{equation}
In \eqref{eq:real_heateqn} $t$ is positive time variable while in
\eqref{eq:real_heat} we have that $t=\sigma^{2}/2$ where $\sigma^{2}$ is the
variance of the corresponding probability distribution.

A main point of the present article is to advocate the following assertion:

\vskip .05in
\noindent
\it If instead of having $x \in \mathbf{R}$, a continuum,  one has that $x \in \mathbf{Z}$, the
discrete space of integers, then the discrete analogue of \eqref{eq:real_heat} is
\begin{equation}
  \label{def:discrete_gaussian}
K_{\mathbb{Z}}(t,x):=e^{-2t}I_{x}(2t)
\end{equation}
where $I_{x}(2t)$ is the $I$-Bessel function defined either by the integral
\begin{equation}\label{eq. Ix defn}
I_{x}(z)=\frac{1}{\pi}\int\limits_{0}^{\pi}e^{z\cos(\theta)}\cos(x\theta)d\theta,
\end{equation}
or by the series representation
  \begin{equation}
    \label{def:Ix}
    I_{x}(z)=\sum_{n=0}^{\infty}\frac{(z/2)^{2n+x}}{n!\Gamma(n+1+x)}
\end{equation}
for $x \geq 0$ and $I_{-x}(z) :=I_{x}(z)$. See Figure \ref{fig:Gaussians}. \rm

\begin{figure}
    \centering
    \includegraphics[width=0.6\linewidth]{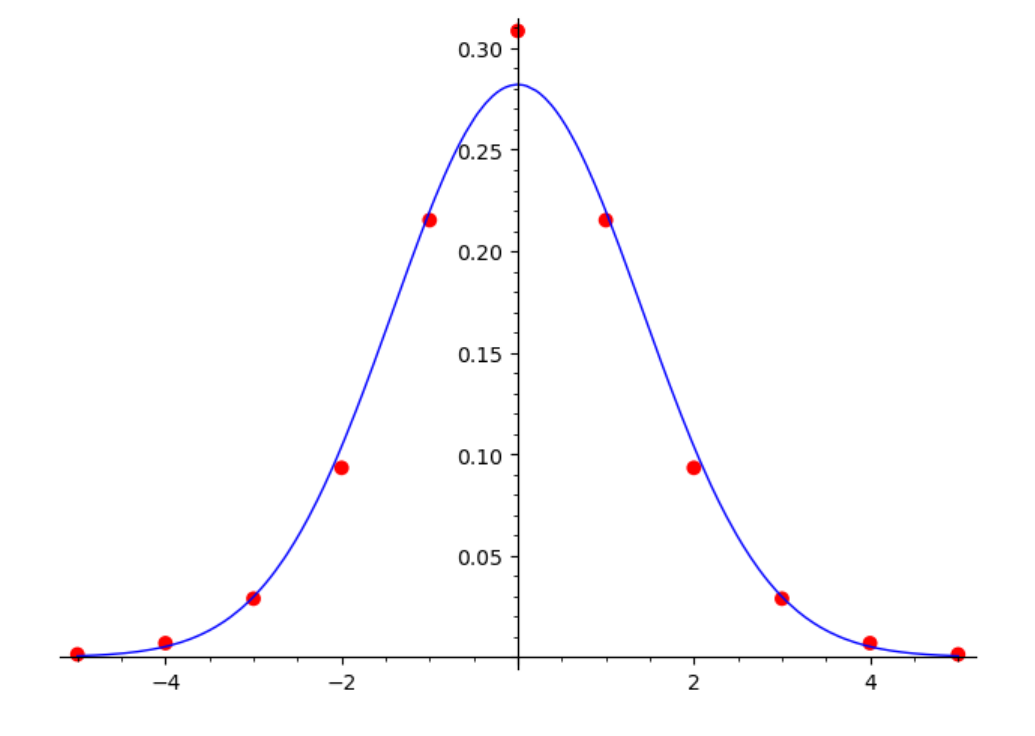}
    \caption{The dots show the discrete Gaussian at $t=1$ as a function of integers $x$ and the curve is the ordinary Gaussian with $t=1$ as a function of real numbers $x$.
}
    \label{fig:Gaussians}
\end{figure}

\vskip .05in
The phrase \it discrete Gaussian \rm is already used in multiple ways in the literature, especially in the field
of lattice-based cryptography, and it often means the probability
density function associated to a $\mathbf{Z}$-valued random variable $W$ for which
\begin{equation}\label{eq:discretized_Gauss}
\text{\rm Prob}(W = n) = c e^{-n^{2}/4t}  \,\,\,\,\,\text{\rm for $n \in \mathbb{Z}$}
\end{equation}
with the normalizing constant $c$ necessarily given by
\begin{equation}\label{eq:normalizing_constant}
c^{-1} = \sum\limits_{k=-\infty}^{\infty}e^{-k^{2}/4t};
\end{equation}
see, for example, \cite{CKS22}.
(The constant \eqref{eq:normalizing_constant}
is known as a Thetanullwerte $\theta(0,i/2t)$ and is an object of extensive mathematical significance; 
see, for example, \cite{Mu83}.)
On the other hand, in \cite{CKS22}, \eqref{eq:discretized_Gauss}
is recognized as a \it discretized \rm  Gaussian; see the discussion before Definition 1.1 in \cite{CKS22}.
Other authors, such as \cite{Li24}, distinguish between \eqref{eq:discretized_Gauss},
which they refer to as the \it sampled Gaussian, \rm and \eqref{def:discrete_gaussian}, which \cite{Li24} calls the discrete analogue of the Gaussian;
see section 2.6 of \cite{Li24}.

\vskip .05in
\it \noindent
Throughout this article, we will refer to \eqref{def:discrete_gaussian} as the discrete Gaussian and to \eqref{eq:discretized_Gauss} either as the discretized or sampled Gaussian.  \rm

\vskip .05in

We now turn to discussing the many ways 
the discrete Gaussian %\eqref{def:discrete_gaussian} %and its variants
satisfies structural results which are analogous
to those which are fulfilled by  \eqref{eq:real_heat}.

\vskip .05in
One immediate justification for calling \eqref{def:discrete_gaussian}  the discrete Gaussian
is that it is the fundamental solution to the heat equation on $\mathbb{Z}$.
More precisely, for  $x\in\mathbb{Z}$  and $t \in \mathbf{R}^{+}$, we have that
\begin{equation}
  \label{eq:deltaKZ}
\left(\Delta_{\mathbb{Z}}+\frac{\partial}{\partial t}\right)K_{\mathbb{Z}}(t,x)=0
\,\,\,\,\,
\text{\rm with}
\,\,\,\,\,
K_{\mathbb{Z}}(0,x)=\delta_{0}(x)
\end{equation}
and where
$$
\Delta_{\mathbb{Z}}f(x)=2f(x)-f(x+1)-f(x-1).
$$
As far as we know, \eqref{eq:discretized_Gauss} does not satisfy such an elementary differential
or difference equation.
We note that operator $\Delta_{\mathbb{Z}}$ appears in various
physics contexts,
such as the study of discrete random Schr\"odinger operators and Anderson
localization \cite{An58}. We call the operator $\Delta_{\mathbb{Z}}$
a discrete, or combinatorial, Laplacian because it can be defined in
analogy with the usual Laplacian
$$
\Delta_{\mathbb R}=-\frac{\partial^2}{\partial x^2}%+\frac{\partial^2}{\partial y^2};
$$
see Section \ref{sec:The-building-block} below.  Moreover, $\Delta_{\mathbb{Z}}$ is in fact a discretization of $\Delta_{\mathbb R}$.

The claim that $e^{-2t}I_{x}(2t)$
satisfies the heat equation \eqref{eq:deltaKZ} follows from the basic relation
\[
I_{x+1}(t)+I_{x-1}(t)=2\frac{d}{dt}I_{x}(t),
\]
which in turn can easily be derived from the definitions (\ref{eq. Ix defn}) or (\ref{def:Ix}).
The solution of \eqref{eq:deltaKZ} has been discovered and rediscovered for many years;
see for example \cite{Ba49, Fe66,GI02,KN06}.  Nonetheless, we feel it bears
further emphasis, particularly in light of recently emerging applications, some
of which we review in this survey.

Let us first point out the following rescaled convergence, which makes precise the intuitive statement that the
discrete heat kernel approaches the continuous heat kernel as the discretization becomes arbitrarily fine.

\begin{prop}
\label{prop:discretecontinuos} Let $\{\alpha_{n}\}$ be any sequence of integers indexed by $n \in \mathbb{Z}^{+}$
such that
$$
\alpha_{n}/n \rightarrow x \in \mathbb{R}
\,\,\,\,\,
\text{\rm as}
\,\,\,\,\,
n \rightarrow \infty.
$$
Then
\[
  \lim_{n\to \infty}nK_{\mathbb Z}(n^2t, \alpha_{n})=
  \lim_{n\rightarrow\infty}ne^{-2n^{2}t}I_{\alpha_{n}}(2n^{2}t)=\frac{1}{\sqrt{4\pi t}}e^{-x^{2}/4t}.
\]
\end{prop}

Proposition \ref{prop:discretecontinuos}  was proved by K. Athreya in \cite{At87}; we provide a different proof in Section \ref{sec: local limit thms}.  
Alternative asymptotic expressions can be
found in Pang \cite{Pa99} and Cowling-Meda-Setti \cite{CMS00}.  Both of these works show how estimates for
the heat kernel on ${\mathbb Z}$ can be used to derive estimates for heat kernels on more general regular graphs and trees, see section \ref{sec:The-building-block}.

Note here the following
philosophical point.  While it is understood from numerical analysis
that for a compact space each discrete eigenvalue converges after rescaling
to the corresponding continuous one when the mesh size goes to zero,
there seems to be no general uniformity result associated to this convergence. Indeed, the discrete problem
has only a finite number of eigenvalues while the continuous one has
an infinitude (or for unbounded spaces the spectrum may even be purely
continuous). The above proposition therefore shows a significant advantage by using
the discrete heat kernel which will yield an important application in section
\ref{sec:Application-3:-asymptotics} below.  In general terms, the rescaled discrete Gaussian is a good way
of packaging all the spectral information when
passing from a lattice to a continuum.

In statistics, as noted in \cite{DH13}, practitioners often face a significant inaccuracy when approximating integer-valued ("lattice-valued") random variables by the real Gaussian suggested by the central limit theorem. The following theorem proved in section \ref{sec: local limit thms} responds to this general problem by instead proposing an approximation with the discrete Gaussian:

\begin{thm} %\label{thm. discrete CLT}
Let $\{X_k\}_{k=1}^{\infty}$ be a sequence of  independent, identically distributed integer-valued random variables with
finite mean $\mu$ and finite variance $\sigma^{2}$, and assume that $\sigma^2>|\mu |$.
Assume further there is no infinite subprogression $a+\ell\mathbb{Z}$ of $\mathbb{Z}$ with $\ell>1$
such that $X_1$ takes values in $a+\ell\mathbb{Z}$ almost surely. Let $S_n:= \sum_{k=1}^n X_k$. Then
$$
\sup_{m\in\mathbb{Z}}\sqrt{n}\left| \mathbf{P}(S_n=m) -
\left(\frac{\sigma^2+\mu}{\sigma^2-\mu}\right)^{m/2} e^{-n\sigma^2}I_m(n\sqrt{\sigma^4-\mu^2}) \right|\to 0
\,\,\,\,
\text{as $n\to\infty$.}
$$
\end{thm}

Spectral zeta functions are relevant in particular for physics, see \cite{El12} and further comments in section \ref{sec:spectral}. Their special values for graphs have received some recent interest in number theory \cite{FK17, XZZ22, KP23}. We obtain new explicit values for the spectral zeta function of Bethe lattices, i.e. $(q+1)$-regular trees $T_{q+1}$:

\begin{thm}
For any $q \geq 1$, $\zeta_{T_{q+1}}(0)=1$, $\zeta_{T_{q+1}}(-1)=q+1$ and for integers $m>1$,
\begin{equation}
\zeta_{T_{q+1}}(-m)=\sum_{k=0}^m\binom{m}{k}^2q^{m-k} - (q-1)\sum_{j=1}^{\lfloor m/2\rfloor}\sum_{k=0}^{m-2j}\binom{m}{k}\binom{m}{2j+k} q^{m-2j-k}.
\end{equation}
Moreover,
$$  \mathrm{det'}\Delta_{T_{q+1}}:=\exp (-\zeta_{T_{q+1}}'(0)) =\left\{
                                                              \begin{array}{ll}
                                                                (1-q^{-2})^{(1-q)/2}, & \text {when  } q>1 \\
                                                                1, & \text{when  } q=1;
                                                              \end{array}
                                                            \right.
 $$
\end{thm}

To close the introduction we highlight some key features of the discrete Gaussian
\eqref{def:discrete_gaussian} which we explore more thoroughly below.
\begin{itemize}
\item (Section \ref{sec:The-building-block})  The discrete Gaussian is not just the heat kernel on $\mathbb Z$,
but in fact it serves as the building block for heat kernels on all regular graphs.
  \item (Section \ref{sec: local limit thms})  The discrete Gaussian appears in  the limiting density function in the
  local central limit law for a sequence of identically distributed integer-valued random variables with mean zero and
  finite variance
\item (Section \ref{sec:Bessel}) Physical and probabalistic principles apply to the discrete Gaussian which then
motivate and make obvious a number of Bessel function identities.
\item (Section \ref{sec:trig}) An application of the Laplace transform to appropriate functionals of
the discrete Gaussian leads to explicit
evaluation of certain finite trigonometric sums appearing in physics
contexts such as the Verlinde formulas, resistance in electrical networks
or chiral Potts model.
  \item (Section \ref{sec:Application-3:-asymptotics}) The discrete Gaussian
        appears in the study of asymptotics of the determinant of the Laplacians
        in lattices with periodic boundary conditions, of relevance in
        statistical physics and quantum field theory.
\item (Section \ref{sec:spectral})  In general, the spectral zeta function arises from a Mellin transform of the
heat kernel; when considering the discrete Gaussian one obtains a number of identities involving known
zeta functions, and some interesting new computations as well.
\item (Section \ref{sec:diffprivacy}) The field of Differential Privacy is a mathematically rigorous
methodology by which one measures the output of fixed algorithm $\mathcal{M}$ on a given dataset $\mathcal{D}$
with the goal of preventing one from determining any specific entry in $\mathcal{D}$ through repeated use of $\mathcal{M}$;
see \cite{CKS22} and references therein.  In Section \ref{sec:diffprivacy} we will comment on the potential usefulness
of the discrete Gaussian \eqref{def:discrete_gaussian}, as opposed to the discretized Gaussian \eqref{eq:discretized_Gauss},
in the analysis in \cite{CKS22}.  Going further, we will point out that the analysis of Section \ref{sec: local limit thms}
admits an immediate extension to random variables whose values lie in higher dimensional lattices; see, for
example, \cite{AA19}.

\end{itemize}

We think the examples listed above are already a noteworthy variety of applications of the
discrete Gaussian, and we are confident that there are more to come.

\section{The building block of heat kernels of graphs\label{sec:The-building-block}}

Let $X=(VX,EX)$ be a countable graph of bounded degree, where $VX$ and $EX$ denote the sets of vertices and edges, respectively. The \emph{Laplacian} is an operator on $L^{2}(VX,\mathbb{C})$ defined by
\begin{equation}\label{eq:Laplacian_graph}
(\Delta_{X}f)(x)=\sum_{y\sim x}\left(f(x)-f(y)\right),
\end{equation}
where the sum is taken over all vertices adjacent to $x$ in $X$.

A graph is called $(q+1)$-\emph{regular} or \emph{regular} if every vertex has degree $q+1$. The reason for the normalization
of ``$q+1$'' comes form number theory, and we adopt it here as well, because various formulas
become slightly more compact when using this convention.
Chung and Yau observed in \cite{CY99} that the heat kernel on a regular graph can be expressed in terms of the heat kernel on a regular tree.
Building on this, Chinta-Jorgenson-Karlsson \cite{CJK15} proved that the $I$-Bessel
function serves as a crucial ingredient in constructing the heat kernel on any regular graph, similar to the role of the Gaussian $\frac{1}{\sqrt{4\pi t}}e^{-x^{2}/4t}$ in the
heat kernel on homogeneous manifolds. In a related result, Gr\"unbaum
and Iliev point out that $e^{-2t}I_{x}(2t)$ generates,
via Darboux transformations, solutions of more complicated differential-difference
equations on the discrete line, again akin to the Gaussian on the real
line \cite{GI02}.  Let us now be more specific.

Let $T_{q+1}$ be a $(q+1)$-regular tree, also called the \emph{the Bethe lattice}. Chung and Yau
\cite{CY99} give an explicit expression for its heat kernel in a radial
coordinate $r.$ That is, since the heat kernel is obviously radially symmetric,
one first selects a vertex as an origin $0$ and then sets $K_{T_{q+1}}(t,x)=K(t,r)$ where $r$ is the
distance between $0$ and $x$.  In \cite{CY99} it is shown that
\[
K(t,r)=\frac{2e^{-(q+1)t}}{\pi q^{r/2-1}}\int_{0}^{\pi}\frac{\exp\left(2t\sqrt{q}\cos u\right)\sin u(q\sin(r+1)u-\sin(r-1)u)}{(q+1)^{2}-4q\cos^{2}u}du
\]
for $r\geq0$. Cowling-Meda-Setti \cite{CMS00} present a different
formula for $K(t,r)$, namely that
\begin{equation}\label{eq. HK on tree}
K(t,r)=q^{-r/2}e^{-(q+1)t}I_{r}(2\sqrt{q}t)-(q-1)\sum_{j=1}^{\infty}q^{-(r+2j)/2}e^{-(q+1)t}I_{r+2j}(2\sqrt{q}t).
\end{equation}
We note that the normalization of the Laplacian in \cite{CMS00} is different than what we use in this article.
The formula \eqref{eq. HK on tree} was rediscovered by different means in \cite{CJK15}.

One feature of \eqref{eq. HK on tree} is that we can think of
\begin{equation}\label{eq. building block}
q^{-n/2}e^{-(q+1)t}I_{n}(2\sqrt{q}t)
\,\,\,\,\,
\text{\rm for}
\,\,\,\,\,
 n\geq 0
\end{equation}
as a type of building block for the heat kernel on the $(q+1)$-regular tree.  More generally,
in the case of a $(q+1)$-regular graph, we have a similar expression, as given in the following theorem.

\begin{thm}\label{thm:heat_kernel_expansion}
\cite{CJK15} The heat kernel on any $(q+1)$-regular graph $X$ is
$$
K_{X}(t,x)=e^{-(q+1)t}\sum_{m=0}^{\infty}b_{m}(x)q^{-m/2}I_{m}(2\sqrt{q}t),
$$
where $b_{m}(x)=c_{m}(x)-(q-1)(c_{m-2}(x)+c_{m-4}(x)+...)$ and $c_{m}(x)$
is the number of geodesics from the origin to $x$ of length $m\geq0$.
\end{thm}

It is interesting to compare Theorem \ref{thm:heat_kernel_expansion} with known
results in the setting of Riemannian geometry, where the heat kernel $K_{M}(t,x,y)$
on an $n$-dimensional smooth manifold $M$ with smooth metric has an asymptotic expansion of
the form
$$
K_{M}(t,x,y) = \frac{e^{-d_{M}^{2}(x,y)/4t}}{(4\pi t)^{n/2}} \left(\sum\limits_{j=0}^{k}u_{j}(x,y)t^{j} + O(t^{k+1})\right)
\,\,\,\,\,
\text{\rm as $t \rightarrow 0$}
$$
for any integer $k>n/2+2$, points $x,y\in M$, distance function $d_{M}$ on $M$, and computable functions
$\{u_{j}\}$; see, for example, page 152 of \cite{Ch84}.  In certain cases when $M$ is a non-compact symmetric space,
the heat kernel can be written as
$$
F(r)\cdot e^{-at}\cdot\frac{1}{(4\pi t)^{d/2}}e^{-n^{2}/4t}
\,\,\,\,\,
\text{\rm where $r=d_{M}(x,y)$}
$$
and $F(r)$ is an explicitly computable power series; see, for example, \cite{GN98} in the case of hyperbolic spaces.

Continuing with this theme, let us define the asymmetric
Laplacian on $\mathbb{Z}$ associated to $L^2$ functions $f:\mathbb Z\to\mathbb C$ by
\[
\Delta_{p,q} f(x)=(p+q)f(x)-pf(x+1)-qf(x-1)
\,\,\,\,\,
\text{\rm for $p,q>0$ and $x\in\mathbb{Z}$.}
\]
The corresponding fundamental solution, or heat kernel, can be found in \cite{Fe66}
for $p+q=1$, and is given by
% (and rediscovered in \cite{CJK15}),
\begin{equation}\label{eq:pq_heat_kernel}
K_{p,q}(t,x)=\left(\frac{p}{q}\right)^{x/2}e^{-(p+q)t}I_{x}(2\sqrt{pq}t).
\end{equation}
Upon taking $q\in\mathbb{N}$ with $x=n\geq 0$ and $p=1$, we see that $K_{p,q}(t,x)$ coincides with
the ``building block'' of heat kernels on $(q+1)$-regular graphs.

\section{Discrete local limit theorem}\label{sec: local limit thms}

Let us continue with the setting of the last paragraph with $p,q>0$ and $p+q=1$.  Then the heat kernel $K_{p,q}(2t,x)$ with
 $x\in\mathbb{Z}$ can be viewed as the probability distribution of the continuous-time random walk on $\mathbb Z$
 which steps one unit to the left (resp. right) with probability $p$ (resp. $q=1-p$.) Further, the time interval between
 steps has an exponential distribution with density $2e^{-2t}$,
then the function $K_{1/2,1/2}(2t,x)=e^{-2t}I_x(2t)$ is the heat
kernel on $\mathbb{Z}$.

Let us denote by $Y_{p,q,t}$ the integer-valued random variable such that
\begin{equation}
  \label{eq:Ypqt}
  \mathbf{P}(Y_{p,q,t}=m)=K_{p,q}(t,m), \quad m\in\mathbb{Z}.
\end{equation}
The characteristic function $\varphi_{Y_{p,q,t}}$ of $Y_{p,q,t}$ can be computed using the explicit formula for the generating function for
Bessel functions.  Indeed, from \cite[Formula (7.8)]{Fe70} we have that
\begin{align}\label{eq:char_function}\notag
\varphi_{Y_{p,q,t}}(y)&= e^{-(p+q)t}\sum_{x=-\infty}^{\infty} I_x(2\sqrt{pq}t)\left(\sqrt{\frac{p}{q}}e^{iy}\right)^x
\\&= e^{-(p+q)t(1-\cos y)}e^{i(p-q)t\sin y}
\,\,\,\,\,
\text{\rm for $y\in\mathbb{R}$.}
\end{align}

In the following theorem, we prove an analogue of the Central Limit Theorem but with an error bound, or rather a type Berry-Esseen Theorem,
which asserts that under reasonably general conditions the limiting probability distribution for large $n$ for the sum
$S_n:= X_1+\ldots +X_n$ of independent identically distributed integer-valued random variables is
given by $K_{p,q}(t,x)$, for certain $p$, $q$ and $t$ which depend on the distribution of $X_{1}$.

\begin{thm}\label{thm. discrete CLT}
Let $\{X_k\}_{k=1}^{\infty}$ be a sequence of  independent, identically distributed integer-valued random variables with
finite mean $\mu$ and finite variance $\sigma^{2}$, and assume that $\sigma^2>|\mu |$.
Assume further there is no infinite subprogression $a+\ell\mathbb{Z}$ of $\mathbb{Z}$ with $\ell>1$
such that $X_1$ takes values in $a+\ell\mathbb{Z}$ almost surely. Let $S_n:= \sum_{k=1}^n X_k$. Then
$$
\sup_{m\in\mathbb{Z}}\sqrt{n}\left| \mathbf{P}(S_n=m) -
\left(\frac{\sigma^2+\mu}{\sigma^2-\mu}\right)^{m/2} e^{-n\sigma^2}I_m(n\sqrt{\sigma^4-\mu^2}) \right|\to 0
\,\,\,\,
\text{as $n\to\infty$.}
$$
\end{thm}

\begin{proof}
For positive integers $n$, let $Y_n$ denote the integer valued random variable $Y_{p,q,t}$ of (\ref{eq:Ypqt}) with $p+q=1$,
$t=n\sigma^2$ and $(p-q)t=n\mu$.  The assumption that $\sigma^2>|\mu |$ implies that both $p$ and $q$ are non-negative.  Upon
substituting into \eqref{eq:pq_heat_kernel}, we get that
\begin{equation}\label{eq. Y-n dist}
\mathbf{P}\left(Y_n=m\right)= \left(\frac{\sigma^2+\mu}{\sigma^2-\mu}\right)^{m/2} e^{-n\sigma^2}I_m(n\sqrt{\sigma^4-\mu^2}).
\end{equation}
Further, the characteristic function $\varphi_{Y_n}$ is
$$
\varphi_{Y_n}(y)=e^{-2n\sigma^2 \sin^2(y/2)}e^{i\mu n \sin y}
\,\,\,\,\,
\text{\rm for $y\in\mathbb{R}$;}
$$
hence,
$$
 \mathbf{P}(Y_n=m)=\frac{1}{2\pi}\int\limits_{-\pi}^{\pi}e^{-2n\sigma^2\sin^2(\theta/2)}e^{i\mu n \sin \theta}e^{-im\theta}d\theta.
$$

Set $X=X_1$, and let $\varphi_{X}(\theta)=e^{i\mu \theta}\mathbb{E}(e^{i\theta( X-\mu)})= e^{i\mu \theta}\varphi_{X-\mu}(\theta)$ be the characteristic function of $X$. Then we have that
$$
 \mathbf{P}(S_n=m)=\frac{1}{2\pi}\int\limits_{-\pi}^{\pi}\left(\varphi_{X-\mu}(\theta)\right)^n e^{i\mu n \theta} e^{-im\theta}d\theta.
$$
Combining the above two displayed equation, we conclude that
\begin{equation}\label{eq. disrt of difference}
\mathbf{P}(S_n=m)-\mathbf{P}(Y_n=m)= \frac{1}{2\pi}\int\limits_{-\pi}^{\pi}\left[\left(\varphi_{X-\mu}(\theta)\right)^n -e^{-2n\sigma^2\sin^2(\theta/2)}e^{i\mu n ( \sin \theta - \theta)} \right] e^{i\mu n \theta} e^{-im\theta}d\theta.
\end{equation}
A simple change of variables in \eqref{eq. disrt of difference} yields that
\begin{multline}\label{eq. disrt of abs difference}
\sqrt{n}\left|\mathbf{P}(S_n=m)-\mathbf{P}(Y_n=m)\right| \\ \leq \frac{1}{2\pi}\int\limits_{-\pi\sqrt{n}}^{\pi\sqrt{n}}\left|\left(\varphi_{X-\mu}(\theta/\sqrt{n})\right)^n -e^{-2n\sigma^2\sin^2(\theta/(2\sqrt{n}))} e^{i\mu n ( \sin( \theta/\sqrt{n} )- \theta/\sqrt{n})} \right| d\theta.
\end{multline}
For any $\theta\in\mathbb{R}$, one can use the Taylor series expansions of $\varphi_{X-\mu}(\theta)$ and $\sin \theta$ in the neighborhood
of $\theta=0$ to deduce that
$$
\lim_{n\to\infty}  \left( \left(\varphi_{X-\mu}(\theta/\sqrt{n})\right)^n -e^{-2n\sigma^2\sin^2(\theta/(2\sqrt{n}))} e^{i\mu n ( \sin( \theta/\sqrt{n} )- \theta/\sqrt{n})}\right) =0.
$$
If one were to substitute \eqref{eq. Y-n dist} into \eqref{eq. disrt of abs difference},
the proof would follow if one could apply the Lebesgue dominated convergence theorem in \eqref{eq. disrt of abs difference} and
conclude that the right-hand side of \eqref{eq. disrt of abs difference} is $o(1)$ as $n\to\infty$.  In other words, it suffices to show that the function
$$
\left|\left(\varphi_{X-\mu}(\theta/\sqrt{n})\right)^n -e^{-2n\sigma^2\sin^2(\theta/(2\sqrt{n}))} e^{i\mu n ( \sin( \theta/\sqrt{n} )- \theta/\sqrt{n})} \right|\chi_{[-\pi\sqrt{n},\pi\sqrt{n}]}(\theta)
$$
is dominated by an integrable function on $\mathbb{R}$.  Let us establish this condition now.

By the Taylor series expansion of $\varphi_{X-\mu}(\theta)$ and $e^{-2\sigma^2\sin^2(\theta/2)}$  in the neighborhood of $\theta=0$, we have that
$$
|\varphi_{X-\mu}(\theta)|\leq |1-\frac{\theta^2\sigma^2}{4}|
\,\,\,\,\,
\text{\rm and}
\,\,\,\,\,
e^{-2\sigma^2\sin^2(\theta/2)} \leq |1-\frac{\theta^2\sigma^2}{4}|.
$$
Hence, there exists some $\delta>0$ such that for all $\theta$ such that $|\theta|\leq \delta\sqrt{n}$ one has that
$$
\left|\left(\varphi_{X-\mu}(\theta/\sqrt{n})\right)^n -e^{-2n\sigma^2\sin^2(\theta/(2\sqrt{n}))} e^{i\mu n ( \sin( \theta/\sqrt{n} )- \theta/\sqrt{n})} \right|\leq 2e^{-\sigma^2\theta^2/4}.
$$
Since the random variable $X$ does not take values in an arithmetic subprogression, we have that
$|\varphi_{X-\mu}(\theta)| \neq 1$  whenever $\theta \neq 0$. Indeed, if there were such  $\theta \neq 0$, then by the triangle inequality  $e^{i\theta(X-\mu)}$ had to be constant, or deterministic, as a random variable. This contradicts the subprogression hypothesis.
Therefore, $|\varphi_{X-\mu}(\theta)|<1$ for $0<|\theta|\leq \pi$.  By continuity, and hence uniform continuity on closed intervals,
we deduce that there exists a constant $0<C_X<1$ such that $|\varphi_{X-\mu}(\theta/\sqrt{n})|\leq C_X$ for all $\delta\sqrt{n}\leq |\theta|
\leq \pi\sqrt{n}$.

The inequality $\sin^2(\theta/(2\sqrt{n}))\geq \theta^2/(\pi n)$ which holds true for $\theta\in (-\pi\sqrt{n},\pi\sqrt{n})$, so then
$$e^{-2n\sigma^2\sin^2(\theta/(2\sqrt{n}))} \leq e^{-2\sigma^2 \theta^2/\pi},$$
for all $\delta\sqrt{n}\leq |\theta|\leq \pi\sqrt{n}.$
When combining the above bounds, we get that
\begin{align}\notag
  &\left|\left(\varphi_{X-\mu}(\theta/\sqrt{n})\right)^n -e^{-2n\sigma^2\sin^2(\theta/(2\sqrt{n}))} e^{i\mu n ( \sin( \theta/\sqrt{n} )- \theta/\sqrt{n})} \right| \chi_{[-\pi\sqrt{n},\pi\sqrt{n}]}  \\ & \label{eq:upper_bound} \leq \left(C_X^n + e^{-2\sigma^2 \theta^2/\pi}\right)
  \chi_{[-\pi\sqrt{n},\pi\sqrt{n}]}(\theta) +  2e^{-\sigma^2\theta^2/4}.
\end{align}
The right-hand-side of \eqref{eq:upper_bound} is an integrable function on $\mathbb{R}$, uniformly in $n$, so
then the proof of the theorem is complete.
\end{proof}

In the special case when the mean $\mu=\mathbf{E}(X)=0$,  we get the following corollary.

\begin{corollary}\it
In addition to the assumptions as stated in Theorem \ref{thm. discrete CLT}, assume that $\mu=0$.
Then
$$
\sup_{m\in\mathbb{Z}}\sqrt{n}\left| \mathbf{P}(S_n=m) - e^{-n\sigma^2}I_m(n\sigma^2) \right|\to 0, \quad \text{as   } n\to\infty.
$$
\end{corollary}
When the set of i.i.d. random variables $\{X_{k}\}$ almost surely takes values in an arithmetic progressions $a+\ell\mathbb{Z}$ with $\ell>1$,
one can prove stronger variants of local limit theorems; see for example the classical books \cite{IL71} or \cite{Pe75} as well
as the recent survey \cite{SW23}.

Other local limit theorems when assuming the conditions of Theorem \ref{thm. discrete CLT} take the form
\begin{equation}\label{eq. local limit with cont gauss}
\sup_{m\in\mathbb{Z}}\sqrt{n}\left| \mathbf{P}(S_n=m) - \frac{1}{\sqrt{2\pi n}\sigma} e^{-(m-n\mu)^2/(2n\sigma^2)}\right|
\to 0
\,\,\,\,\,
\text{\rm as $n\to\infty$;}
\end{equation}
see for example  Theorem 13 of \cite{Pe75}).
As noted above,
$$
\sum_{m\in\mathbb{Z}} \frac{1}{\sqrt{2\pi n}\sigma} e^{-m^2/(2n\sigma^2)}  \neq 1,
$$
meaning that the function $\frac{1}{\sqrt{2\pi n}\sigma} e^{-(m-n\mu)^2/(2n\sigma^2)}$ when viewed as a function of variable $m\in\mathbb{Z}$ is not a probability distribution on integers.
However, the heat kernel $K_{p,q}(n\sigma^2,m)$ with $p+q=1$ and $p-q=\mu/\sigma^2$ \textit{is a probability distribution on} $\mathbb{Z}$.
We believe this important point provides significant justification for calling $e^{-2t}I_x(2t)$, $t\geq 0$ with $x\in\mathbb{Z}$ the discrete Gaussian.

Finally, let us note here that Theorem \ref{thm. discrete CLT} combined with the asymptotic \eqref{eq. local limit with cont gauss} for the local limit theorem yields that for $\sigma^2>|\mu|$ one has, by the triangle inequality, that
$$
\sup_{m\in\mathbb{Z}}\sqrt{n}\left| \left(\frac{\sigma^2+\mu}{\sigma^2-\mu}\right)^{m/2} e^{-n\sigma^2}I_m(n\sqrt{\sigma^4-\mu^2})  - \frac{1}{\sqrt{2\pi n}\sigma} e^{-(m-n\mu)/2n\sigma^2}\right|\to 0
\,\,\,\,\,
\text{as $n\to\infty$.}
$$
By taking $\mu=0$ in the above inequality, one obtains a alternate proof of Proposition \ref{prop:discretecontinuos}.

\section{Application 1: Bessel identities}
\label{sec:Bessel}

Bessel functions are ubiquitous in mathematical physics, perhaps first appearing
as a solution to Bessel's differential equations arising from the Laplace
equation with cylindrical symmetry. A classic text on Bessel functions is
Watson's book \cite{Wa44}. The fact that the $I$-Bessel function is essentially
a heat kernel makes many of the classical identities immediate from physical
principles. Note that we continue to view $I_{x}(t)$ as a function in two
variables, where the variable $t$ is continuous and where the variable $x$ is discrete. 
We will now illustrate how to
obtain some known Bessel function identities, following \cite{KN06,K12,CJKS23}. In fact, a very nice and much
earlier discussion of this type can be found in \cite{Fe66}.  Though none of the Bessel 
function identities in this section is new,  we wish to illustrate how our perspective 
informs our understanding of these identities.

To begin with, start with a unit
amount of heat at time $t=0$.  The principle of the conservation of heat implies that though the heat diffuses over the integers as time passes, the total amount of heat at all
points sums to one at any time.  Mathematically, this principle immediately implies that
\begin{equation}\label{eq:heat_conservation}
\sum_{j=-\infty}^{\infty}e^{-2t}I_{j}(2t)=1,
\end{equation}
for all $t\geq 0.$
We will deduce this identity in an alternate way in an example below. If
we change the continuous variable from real $t$ to $z$, we can rewrite \eqref{eq:heat_conservation} as
\begin{equation}\label{eq:heat_conservation_complex}
\sum_{j=-\infty}^{\infty}I_{j}(z)=e^{z}.
\end{equation}
Since \eqref{eq:heat_conservation_complex} holds for any complex $z$, the by the principle of analytic
continuation applies.  In particular, we can rephrase \eqref{eq:heat_conservation_complex} in terms of 
the $J$-Bessel $J_{n}(z)$ to get that
\[
\sum_{j=-\infty}^{\infty}i^{-j}J_{j}(z)=e^{iz},
\]
which leads to
\[
J_{0}(z)+2\sum_{k=1}^{\infty}(-1)^{k}J_{2k}(z)=\cos(z)
\]
and
\[
\sum_{k=0}^{\infty}(-1)^{k}J_{2k+1}(z)=\frac{1}{2}\sin(z).
\]

By applying the same reasoning as above to the heat kernel of the asymmetric
Laplacian (corresponding to diffusion with a drift) as in the previous section, one gets more generally that
\begin{equation}\label{eq:generating_series}
\sum_{j=-\infty}^{\infty}J_{j}(z)x^{j}=e^{\frac{z}{2}(x-x^{-1})}.
\end{equation}
This generating series identity \eqref{eq:generating_series} is sometimes known
as the \emph{Schl\"omilch formula,} and it is also given in \cite{Fe66}, though derived
through very different means.  

Solutions to differential equations have a semi-group property in
time; in probability theory, such an identity is the Chapman-Kolmogorov equation.
In our case, the resulting formula is that 
\[
J_{n}(t+s)=\sum_{k=-\infty}^{\infty}J_{n-k}(t)J_{k}(s).
\]
This is called \emph{Neumann's identity} for the $J$-Bessel function and
is deduced in a very different way in \cite{Wa44}.

As in \cite{KN06}, we can write the heat kernel on $\mathbb{Z}/n\mathbb{Z}$ in two ways,
once through its spectral expansion and once by periodizing the heat kernel on $\mathbb{Z}$, also
known as the method of images.  Since the heat kernel is unique, the two expressions are equal.
As such, we obtain a discrete analogue of the Poisson summation formula, namely that
\begin{equation}\label{eq:discrete_Poisson}
\sum_{j=-\infty}^{\infty}e^{-2t}I_{x+jn}(2t)=\frac{1}{n}\sum_{k=0}^{n-1}e^{-4\sin^{2}(k\pi/n)t}e^{2\pi ikx/n}.
\end{equation}
We refer to \cite{ADG02} for a derivation when $x=0$ without the heat kernel perspective.
In the case when $n=1$ and $x=0$, \eqref{eq:discrete_Poisson} gives that
\[
\sum_{j=-\infty}^{\infty}e^{-2t}I_{j}(2t)=1,
\]
thus reproving \eqref{eq:heat_conservation}. That is, the consideration
of heat diffusion on the space consisting of exactly one point gives
rise to a nontrivial Bessel function summation identity.

We conclude this section with some examples of Bessel convolution identities.
Define the convolution of two functions on $(0,\infty)$ by
\[
f*g(x)=\int_{0}^{x}f(t)g(x-t)dt.
\]
Continuing to use the heat kernel formalism as above, we can deduce that
\[
J_{l}*J_{n}(x)=(-1)^{\frac{l+n+1}{2}}\left(\cos x+J_{0}(x)-2\sum_{k=0}^{(l+n-1)/2}(-1)^{k}J_{2k}(x)\right)
\]
 provided the sum $l+n$ of the positive integers $l$ and $n$ is odd; see \cite{K12}
for details. For
example, with $k=0$ and $n=1$, one gets that
\[
\int_{0}^{x}J_{0}(t)J_{1}(x-t)dt=J_{0}(x)-\cos(x).
\]
Further integral formulas are deduced in \cite{K12} such as
\[
\int_{0}^{x}J_{2l+1}(t)dt=1-J_{0}(x)-2\sum_{k=1}^{l}J_{2k}(x).
\]
%Less classical, probably not appearing in the literature can be obtain
%when equating various formulas for a certain heat kernel (section
%\ref{sec:The-building-block} should give rise to a large source of
%such Bessel identities).

As another example, for positive integers $n$ and $m$, one has that
\[
I_{n+m}(t)=\sum_{k=0}^{\infty}\frac{(-1)^{k}}{2^{k+1}}(I_{1}^{*k}*I_{n-1}*I_{m})(t),
\]
where $*k$ denotes the $k$-fold convolution.  This identity is proved in \cite{CJKS23}, 
to which we refer for further discussion and context.  

\section{Application 2: Finite trigonometric sums}
\label{sec:trig}

Finite sums of powers of trigonometric functions of multiple angles,
such as
\[
\sum_{j=1}^{m-1}\frac{e^{2\pi i rj/m}}{\sin^{n}(j\pi/m)}
\,\,\,\,\,
\text{\rm for $m\in\mathbb{N}$ and $r \in\{0,1,\ldots,m-1\}$}
\]
appear in many contexts, including Verlinde formulas, chiral Potts
models, expressions for resistance in electrical network, and modeling angles in proteins and circular
genomes; see \cite{Do92,JKS23} for references.  Let us now describe the methodology developed
in \cite{JKS23} to evaluate these sums, which is based on analysis of the resolvent kernel as derived
from the heat kernel on the discrete torus.  

For $r$ and $m$ as above and $\beta \in \mathbb{R}\setminus \mathbb{Z}$, define
\[
C_{m,r}(\beta,n)=\frac{1}{m}\sum_{j=0}^{m-1}\frac{1}{\sin^{2n}((j+\beta)\pi/m)}e^{2\pi irj/m},
\]
as well as the generating function
\begin{equation}\label{eq:generating_function}
f(s,\beta,r):=\sum_{n=0}^{\infty}C_{m,r}(\beta,n+1)s^{n}
\end{equation}
of the sequence $\{C_{m,r}(\beta,n)\}_{n=1}^{\infty}$.
As shown in \cite{JKS23}, the series \eqref{eq:generating_function}
is related to the resolvent kernel, or the Green's function, with an additive twist on the discrete circle 
$\mathbb{Z}/m\mathbb{Z}$ when viewed as a Cayley graph $C_m$ of degree two.  
Specifically, from the spectral theory of the circle graph $C_m$, the resolvent kernel, 
evaluated at two points $x,y\in C_m$ such that $r=|x-y|\, (\mathrm{mod}\, m)$ can be expressed as
$$
G_{m,\beta}(r,s)= \frac{1}{m}\sum_{j=0}^{m-1} \frac{1}{s+2\sin^2\left(\pi\frac{j+\beta}{m} \right)}
\exp\left(2\pi i \frac{j+\beta}{m}r\right),
$$
for all complex $s$ for which the left-hand side is well-defined.  Hence
$$
f(s,\beta,r)= 2G_{m,\beta}(r,-2s)e^{-2\pi i\beta r/m}.
$$
On the other hand, the twisted resolvent kernel, for sufficiently large Re$(s)$ is the Laplace transform of the twisted heat kernel, which, by the method of imaging, can be expressed as a rapidly convergent sum of discrete Gaussians $e^{-t}I_x(t)$:
\begin{equation*} \label{eq: explicit m,b,beta twisted HK}
K_{X_{m},\chi_{\beta}}(r;t) = \sum_{k\in \mathbb{Z}} e^{-2\pi i \beta k} e^{-t} I_{r+km}(t).
\end{equation*}

\noindent
\textbf{Note:} In \cite{JKS23} the authors employed a different normalization of the heat kernel
than the one given in \eqref{eq:Laplacian_graph}.  In effect, the edge-difference is
given weight $1/2$ rather than $1$, resulting in a different time scale in the heat kernel.
To be consistent with the formulas in \cite{JKS23}, we will follow the normalization
employed there.

By computing the Laplace transform of the heat kernel \eqref{eq: explicit m,b,beta twisted HK}, the
following result is proved in \cite{JKS23}.

\begin{thm}\label{thm:JKS23}
(\cite{JKS23}) For sufficiently small complex $s$,
\[
\sum_{n=0}^{\infty}C_{m,r}(\beta,n+1)s^{n}=2e^{-2\pi i\beta r/m}\frac{U_{m-r-1}(1-2s)+e^{2\pi i\beta}U_{r-1}(1-2s)}{T_{m}(1-2s)-\cos2\pi\beta},
\]
where $T_{n}$ and $U_{n}$ denote the Chebyshev polynomials of the
first and second kind.
\end{thm}

The right-hand-side of the identity in Theorem \ref{thm:JKS23} is a rational function in $s$ with
well-known coefficients.  Therefore, one can obtain an evaluation for each of the coefficients in 
the right-hand-side.  For example, as shown in \cite{JKS23}, for any $k>0$ we have that
For any $k>0$
\[
\sum_{j=1}^{3k-1}\frac{1}{\sin^{4}(j\pi/3k)}\cos(2\pi j/3)=-\frac{1}{45}\left(39k^{4}+30k^{2}+11\right),
\]
which, as far as we know, is a new evaluation of a classical trigonometric sum.  

Previous results of this type appear by many authors, see
the references in \cite{Do92,JKS23} for a start. We find our approach, ultimately arising from the discrete Gausssian, to be structurally pleasing, and it also gives new and general formulas for the exact evaluation of sums of this type.

To illustrate the universality of this idea, if one starts instead with the untwisted discrete Gaussian, one obtains in the
same way (see \cite{K12}) the celebrated special values
\[
\zeta(k)=\sum_{1}^{\infty}\frac{1}{n^{k}}
\]
for $k$ even, such as $\zeta(2)=\pi^{2}/6$, $\zeta(4)=\pi^{4}/90$
\emph{etc.,} due to Euler.
Furthermore, in \cite{JKS23} we were able to answer a question posed
in \cite{XZZ22} about certain related trigonometric sums, that, thanks
to the intriguing formulas in \cite{XZZ22} linking the discrete and
continuous, provides a new approach to special values of Dirichlet
$L$-functions.

\section{Application 3: Asymptotics by passing from the discrete to the continuous\label{sec:Application-3:-asymptotics}}

Consider a discrete tori in $d$ dimensions, or equivalently, a rectangular box in the standard lattice $\mathbb{Z}^{d}$
equipped with periodic boundary conditions, see Figure \ref{fig:torus}.  The question is how does
the determinant of the (finite) graph Laplacian grows as the rectangular box becomes
larger. This invariant counts the number of rooted spanning trees. (We note that arbitrary tori can also be treated; see \cite{CJK12}.)
In thermodynamics one speaks of the infinite volume limit and the
continuum limit. These are both visible in our asymptotics which extend
what was known in dimension 2 by work of Kasteleyn, Barber, Duplantier-David
\cite{Ka61, Ba70, DD88}. In short, among the key ideas in our approach
are, first, to use a discrete version of the Poisson summation formula
in order to encode the spectral data of the graph in terms of discrete
Gaussians, then second, in the asymptotic analysis to relate the discrete
to the continuous via Proposition \ref{prop:discretecontinuos} above.  The
following is one of the main results from \cite{CJK10, CJK12}.

\begin{thm}
(\cite{CJK10, CJK12}) \label{thm:CJK10}Let $A_{n}$ be a sequence of integral
$d\times d$ matrices with $\det A_{n}\rightarrow\infty$ and $\text{A}_{n}/(\det A_{n})^{1/d}\rightarrow A\in\mathrm{SL}(d,\mathbb{R\textrm{)}}$.
As $n\rightarrow\infty,$
\[
\log\det\Delta_{\mathbb{Z}^{d}/A_{n}\mathbb{Z}^{d}}=\log\det\Delta_{\mathbb{Z}^{d}}\det A_{n}+\frac{2}{d}\log\det A_{n}+\log\det\Delta{}_{\mathbb{R}^{d}/A\mathbb{Z}^{d}}+o(1).
\]
\end{thm}

\begin{figure}
    \centering
    \includegraphics[width=0.3\linewidth]{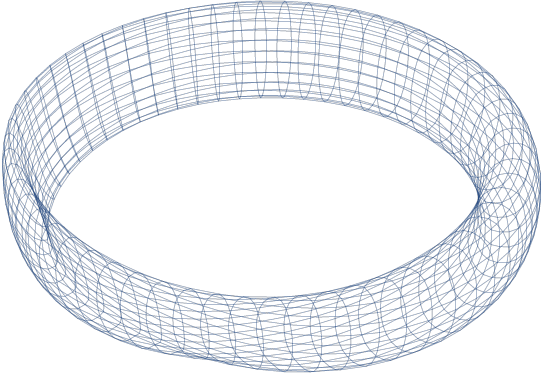}
    \caption{Spectral asymptotics of a discretized torus as the mesh goes to zero}
    \label{fig:torus}
\end{figure}

\noindent In the above theorem, $\Delta_{\mathbb{Z}^{d}/A_{n}\mathbb{Z}^{d}}$ is the Laplacian on the discrete torus; 
$\Delta_{\mathbb{Z}^{d}}$ is the Laplacian on the lattice $\mathbb{Z}^d$ and $\Delta{}_{\mathbb{R}^{d}/A\mathbb{Z}^{d}}$ 
is the Laplacian on the (continuous) torus $\mathbb{R}^{d}/A\mathbb{Z}^{d}$.  The determinant of the Laplacian on the
discrete torus is a finite determinant, in the sense of linear algebra.  However, the determinant of the Laplacian on
the continuous torus is obtained through zeta regularization, which adds an entirely new level of complexity to the matter.  

Let us briefly explain how the discrete Gaussian comes into play in this setting. Following the usual formalism, $\log\det\Delta_{\mathbb{Z}^{d}/A_{n}\mathbb{Z}^{d}}$ is related to the Mellin transform of the trace of the heat kernel.  
The trace of the heat kernel is the theta function, which in our discrete setting involves a sum of products of discrete Gaussians 
$\prod\limits_{j=1}^d e^{-2t}I_{n_jk_j}(2t)$. A careful analysis of the asymptotics of the discrete Gaussian is the key ingredient in the proof of Theorem \ref{thm:CJK10}.

Theorem \ref{thm:CJK10} is of interest in combinatorics, statistical physics and quantum field theory (QFT).  Several recent papers in mathematical physics, notably \cite{HK20,IK22,Gr23}, have extened our results. A problem is raised in \cite{RV15}
in connection with their approach to a combinatorial QFT that Theorem
\ref{thm:CJK10} provides the first answer of (apart from the 2D case
that was already known as just mentioned).

We point out that in higher dimensions $d>2$ the lead terms in the asymptotics of Theorem \ref{thm:CJK10}
are established in \cite{SW00,SS01,Ly05} (see also their bibliographies
for further references). In particular, the formula of Sokal-Starinets
involves the $I$-Bessel function in the same way as in \cite{CJK10},
but they have another way of interpreting its appearance, presumably
not as a discrete Gaussian. Detailed asymptotics for more general
lattice domains in 2D are obtained in the work of Kenyon, see \cite{Ke00}.

For other contexts in physics where recent similar considerations of discrete heat kernels appear, see \cite{KS18, KS23}. 

\section{Spectral zeta functions}
\label{sec:spectral}
\subsection{Spectral zeta function of $\mathbb{Z}$}

The zeta function regularization method in physics started in the
1970s with papers by Dowker-Critchley and by Hawking. For a survey
of applications to mathematical physics of this procedure, see \cite{V03, El12, DEK12}. The zeta function in question is defined by a integral
transform (the Mellin transform) of the heat kernel. The case of the
circle essentially gives rise to the Riemann zeta function. Much less
standard is to take the Mellin integral transform of $e^{-2t}I_{n}(2t)$,
more precisely
\[
\zeta_{\mathbb{Z}}(s)=\frac{1}{\Gamma(s)}\int_{0}^{\infty}e^{-2t}I_{0}(2t)t^{s}\frac{dt}{t},
\]
for $0<\mathrm{Re(s)<1/2}$ and then extend by meromorphic continuation,
as was done in \cite{FK17}. Is also this function relevant for physics?
There are some indications that this the case. First, using the method
of \cite{CJK10} the asymptotics of the spectral zeta function of
discrete tori were determined in \cite{FK17}. Note that in the zeta
function regularization method the determinant of the Laplacian discussed
above corresponds to the special value $\zeta'(0)$ of the corresponding
zeta function.

Second, here is a different context motivated by understanding the
Casimir energy in quantum field theory. The paper \cite{NP00} compared
the spectral zeta functions of a circular disk and the solid cylinder
and obtained the following relationship:
\[
\zeta_{cir}(s)=2\sqrt{\pi}\frac{\Gamma((s+1)/2)}{\Gamma(s/2)}\zeta_{cyl}(s+1).
\]
The factor appearing here is the same independently of whether one
considers Dirichlet or Neumann boundary conditions in the circular
disk and solid cylinder.

It turns out from calculations in \cite{FK17,Du19} that
\[
\zeta_{\mathbb{Z}}(s)=\frac{1}{\sqrt{\pi}4^{s}}\frac{\Gamma(1/2-s)}{\Gamma(1-s)}=\left(\begin{array}{c}
-2s\\
-s
\end{array}\right),
\]
essentially an extension of the Catalan numbers, therefore the factor
appearing in \cite{NP00}
\[
\frac{2\sqrt{\pi}\Gamma((1+s)/2)}{\Gamma(s/2)}
\]
equals
\[
\frac{2^{s}}{\zeta_{\mathbb{Z}}((1-s)/2)}.
\]
 Note that, via the functional symmetry given in \cite{FK17}, $\zeta_{\mathbb{Z}}((1-s)/2)$
is essentially equal to $\zeta_{\mathbb{Z}}(s/2)$ similarly to what
is the case for the Riemann zeta function. More precisely, with $\xi_{\mathbb{Z}}(s):=2^{s}\cos(\pi s/2)\zeta_{\mathbb{Z}}(s/2),$
it holds that
\[
\xi_{\mathbb{Z}}(s)=\xi_{\mathbb{Z}}(1-s)
\]
for all $s\in\mathbb{C}.$ Is there a spectral explanation for the
appearance of $\zeta_{\mathbb{Z}}(s)$ in \cite{NP00} and if so,
how does it generalize?

Finally, we mention the paper \cite{KP23} which argues that special values of $\zeta_{\mathbb{Z}}(s)$ are related to the volume of spheres.

\subsection{Application 4: Spectral zeta function of regular trees and hypergeometric identities}

The spectral zeta function $\zeta_{T_{q+1}}$ on the $(q+1)$-regular tree $T_{q+1}$ was computed in \cite[Theorem 1.4]{FK17} by representing it as an integral over the spectral measure on $T_{q+1}$. Note that the case $q=1$ was discussed in the previous subsection. It was proved that for $q>1$ one has
\begin{equation}\label{eq. spect zeta of tree}
\zeta_{T_{q+1}}(s)= \frac{q(q+1)}{(q-1)^2 (\sqrt{q}-1)^{2s}}F_1\left(\frac{3}{2},s+1,1,3; -\frac{4\sqrt{q}}{ (\sqrt{q}-1)^{2}}, \frac{4\sqrt{q}}{ (\sqrt{q}+1)^{2}}\right).
\end{equation}

On a different front, the spectral zeta function can be expressed in terms of a Mellin transform of the heat kernel on $T_{q+1}$. Namely, for $0<\mathrm{Re}(s)<1/2$ we have that
$$
\zeta_{t_{q+1}}(s)=\frac{1}{\Gamma(s)} \int\limits_{0}^{\infty} K(t,0)t^{s}\frac{dt}{t},
$$
where $K(t,r)$ is the heat kernel on the tree $T_{q+1}$, defined by \eqref{eq. HK on tree}. The Mellin transform of the rescaled discrete Gaussian \eqref{eq. building block} that appears in \eqref{eq. HK on tree} can be evaluated, for any $n\geq 0$ and $0<\mathrm{Re}(s)<1/2$ as follows:
$$
\mathcal{M}\left(q^{-n/2}e^{-(q+1)t}I_{n}(2\sqrt{q}t)\right)(s)=q^{-n/2}\frac{1}{\pi}\int\limits_0^\pi \left(\int\limits_0^{\infty}e^{-(1-2\sqrt{q}\cos\theta+q)t}t^{s-1 }dt\right)\cos\theta n d\theta.
$$
where the application of Fubini-Tonelli theorem is justified by assumptions on $s$ and the fact that $1-2\sqrt{q}\cos\theta+q>0$ for $q>1$.
Therefore,
$$
\mathcal{M}\left(q^{-n/2}e^{-(q+1)t}I_{n}(2\sqrt{q}t)\right)(s)=q^{-n/2}\frac{1}{2\pi}\int\limits_{-\pi}^\pi \frac{\cos\theta n}{(1-2\sqrt{q}\cos\theta+q)^s} d\theta.
$$
Applying a simple manipulation of the above integral in combination with \cite{GR07}, formula 9.112 (defining the Gauss hypergeometric function $F$) yield that
$$
\mathcal{M}\left(q^{-n/2}e^{-(q+1)t}I_{n}(2\sqrt{q}t)\right)(s)=q^{-n-s}\frac{\Gamma(n+s)}{\Gamma(n+1)}F\left(s,s+n;n+1;\frac{1}{q}\right).
$$
Therefore, the zeta function on $T_{q+1}$ can be expressed as
\begin{equation}\label{eq. spect zeta of tree 2}
\zeta_{T_{q+1}}(s)= q^{-s}F\left(s,s,1;\frac{1}{q}\right) - (q-1)\sum_{j=1}^{\infty}q^{-s-2j}\frac{\Gamma(s+2j)}{\Gamma(s)\Gamma(2j+1)}F\left(s,s+2j;2j+1;\frac{1}{q}\right),
\end{equation}
where the termwise integration of the series on the right-hand side of \eqref{eq. HK on tree} is justified by exponential decay of the building blocks \eqref{eq. building block}.
When $q=1$, then
$$
\zeta_{T_{q+1}}(s)=F\left(s,s,1;1\right)=\frac{\Gamma(1-2s)}{\Gamma(1-s)^2}=\zeta_{\mathbb{Z}}(s).
$$

Comparing \eqref{eq. spect zeta of tree} with \eqref{eq. spect zeta of tree 2} yields a new identity for the Appell hypergeometric function $F_1$.

Moreover, the expression \eqref{eq. spect zeta of tree 2} is suitable for computing special values of the spectral zeta function at zero and negative integers. We have the following proposition.
\begin{thm}
For any $q \geq 1$, we have that
$$\zeta_{T_{q+1}}(0)=1 \quad and \quad \zeta_{T_{q+1}}'(0)= \left\{
                                                              \begin{array}{ll}
                                                                \frac{q-1}{2}\log(1-q^{-2}), & \text {when  } q>1 \\
                                                                0, & \text{when  } q=1;
                                                              \end{array}
                                                            \right.
 $$
$\zeta_{T_{q+1}}(-1)=1+q$.  Furthermore, for integers $m\geq 2$, we have that
\begin{equation}\label{eq. zeta at negative int}
\zeta_{T_{q+1}}(-m)=\sum_{k=0}^m\binom{m}{k}^2q^{m-k} - (q-1)\sum_{j=1}^{\lfloor m/2\rfloor}\sum_{k=0}^{m-2j}\binom{m}{k}\binom{m}{2j+k} q^{m-2j-k}.
\end{equation}
\end{thm}
Before we proceed with the proof, let us note an interesting curiosity. Namely, for $q>1$, from \eqref{eq. zeta at negative int} it is obvious that values of $\zeta_{T_{q+1}}$ at negative integers $-m$ are degree $m$ polynomials in $q$. This is reminiscent of the property of the Hurwitz zeta function
$$
\zeta_H(z,q):=\sum_{n=0}^{\infty} \frac{1}{(n+q)^z},\quad \mathrm{Re}(z)>1
$$
that the values of its meromorphic continuation at negative integers $-m$ are $\zeta_H(-m,q)=-\frac{B_{m+1}(q)}{m+1}$, where $B_{m+1}$ are the Bernoulli polynomials (which are degree $m+1$ polynomials in $q$).
\begin{proof}
To evaluate $\zeta_{T_{q+1}}$ and its derivative at $s=0$ we use the following asymptotic expressions as $s\to0$, for a fixed positive integer $q$:
$$
F\left(s,s,1;\frac{1}{q}\right)= 1+ O(s^2),
$$
$$
\frac{\Gamma(s+2j)}{\Gamma(s)\Gamma(2j+1)}= \frac{s}{2j}+O(s^2), \quad  q^{-s-2j}F\left(s,s+2j;2j+1;\frac{1}{q}\right)=q^{-2j}(1+O(s)).
$$
From \eqref{eq. spect zeta of tree 2} we get the following asymptotic expansion as $s\to 0$:
$$
\zeta_{T_{q+1}}(s)=1-(q-1)s\sum_{j=1}^{\infty}\frac{q^{-2j}}{2j} +O(s^2),
$$
which yields that $\zeta_{T_{q+1}}(0)=1$. When $q=1$, trivially $\zeta_{T_{q+1}}'(0)=0$, while for $q>1$ we get
$$
\zeta_{T_{q+1}}'(0)= -\frac{q-1}{2}\sum_{j=1}^{\infty}\frac{(q^{-2})^j}{j}= \frac{q-1}{2}\log(1-q^{-2}).
$$
Evaluation of $\zeta_{T_{q+1}}(-1)$ is trivial, since $\left.\frac{\Gamma(s+2j)}{\Gamma(s)\Gamma(2j+1)}\right|_{s=-1}=0$ for all $j$, hence
$$
\zeta_{T_{q+1}}(-1)=qF\left(-1,-1;1;\frac{1}{q}\right)=1+q.
$$
Finally, let $m\geq 2$. Then, for any $1\leq j\leq \lfloor m/2 \rfloor$ we have
$$
\left.\frac{\Gamma(s+2j)}{\Gamma(s)\Gamma(2j+1)}\right|_{s=-m}=\lim_{z\to 0}\frac{\Gamma(z+2j-m)}{\Gamma(z-m)\Gamma(2j+1)}=\frac{m!}{(2j)!(m-2j)!},
$$
while for $2j>m$ we have
$$
\left.\frac{\Gamma(s+2j)}{\Gamma(s)\Gamma(2j+1)}\right|_{s=-m}=0.
$$
Therefore,
\begin{align}\label{eq. zeta at neg prep} \notag
\zeta_{T_{q+1}}(-m)&=q^m F\left(-m,-m,1;\frac{1}{q}\right)\\& -(q-1) \sum_{j=1}^{\lfloor m/2\rfloor} \frac{m! q^{m-2j}}{(2j)!(m-2j)!} F\left(-m,2j-m;2j+1;\frac{1}{q}\right)
\end{align}
From the definition of the hypergeometric function, it is trivial to deduce that
$$
F\left(-m,-m,1;\frac{1}{q}\right)=\sum_{k=0}^m\binom{m}{k}^2q^{-k}
$$
and that
$$
\frac{m!}{(2j)!(m-2j)!} F\left(-m,2j-m;2j+1;\frac{1}{q}\right)=\sum_{k=0}^{m-2j}\binom{m}{k}\binom{m}{2j+k} q^{-k}.
$$
Inserting this into \eqref{eq. zeta at neg prep} proves \eqref{eq. zeta at negative int}.
\end{proof}

\section{Potential applications to differential privacy}
\label{sec:diffprivacy}

Differential privacy is a mathematically rigorous criterion by which an algorithm that acts
on a dataset can be assessed.  We refer to the article \cite{CKS22} and references
therein for an excellent presentation.  In practice, as described in \cite{CKS22}, one alters
the output of a $\mathbb{Z}$-valued algorithm  $\mathcal{M}$ which produces results from a given dataset $\mathcal{D}$
by adding a random amount.  The main results of \cite{CKS22} are that one can use a discretized
Gaussian random \eqref{eq:discretized_Gauss} (using out terminology) to create algorithms which satisfy differential privacy,
 and that one can effectively and efficiently sample from a discretized Gaussian random variable.

It would be interesting to determine if the discrete Gaussian \eqref{def:discrete_gaussian} also
satisfies the criterion established in \cite{CKS22} for random ``noise'' as required by differential privacy.  
It should be noted that other authors have pointed out the potential usefulness of  \eqref{def:discrete_gaussian}
and distinguished the difference between  \eqref{def:discrete_gaussian} and \eqref{eq:discretized_Gauss};
see section 2.6 of \cite{Li24} who refers to earlier work by the same author in \cite{Li90}.  One can reasonably
expect that, indeed, \eqref{def:discrete_gaussian} will fulfill the conditions needed for differential privacy,
if for no reason other than the triangle inequality, the results from \cite{CKS22} and 
Proposition \ref{prop:discretecontinuos}.  However, it would be better to establish differential privacy
associated to \eqref{def:discrete_gaussian} without appealing to \cite{CKS22}.  We will leave this
study for a later time. 

From the point of view of probability, one can state several advantages that  \eqref{def:discrete_gaussian}
has over \eqref{prop:discretecontinuos}.  First, we note that
the (quite impressive) analysis of \cite{CKS22} includes, for example, estimates for the moment generating function and variance
of the discretized Gaussian \eqref{eq:discretized_Gauss}; see Lemma 8 and Corollary 9 of \cite{CKS22}.  
By contrast, note that the characteristic function of \eqref{def:discrete_gaussian} is explicitly computed;
see \eqref{eq:char_function}.  As such, one has elementary closed-form expressions for all moments of 
\eqref{def:discrete_gaussian}.  See also \cite{AGMP24}.  Second, the sum of two independent discretized Gaussian random variables \eqref{eq:discretized_Gauss}
is close to, but not equal to, a discretized Gaussian random variable; see the comment after Remark 3.7 in \cite{AA19} and 
Theorem 1.1 of \cite{AR24}.  By comparison, the sum of
two independent discrete Gaussians \eqref{def:discrete_gaussian} is a discrete Gaussian, and the convolution of the
probability density functions is, in effect, computed in the above-cited Chapman-Kolmogorov equation; see Section \ref{sec:Bessel}.  
Finally, let us point out that one application of Theorem \ref{thm. discrete CLT} is a means by which one can sample from
a discrete Gaussian is a manner similar to a sampling method for a continuous Gaussian random variable.  In doing so, one
has, at least the beginning of, an analogue to the sampling algorithms developed in \cite{CKS22}.  Other sampling methods
could be developed based on numerical evaluations of Bessel functions; to this end, we found the articles \cite{Am74}
and \cite{GS81} to be interesting.  

Finally, let us note that the articles \cite{AA19} and \cite{AR24} consider discretized Gaussian random variables on
any $d$-dimensional lattice $\Lambda$ in $\mathbb{R}^{d}$.  Certainly, our point of view is amenable to such a setting; see, 
for example, \cite{CJK10} and more generally \cite{CJK12}.  The probability density function is simply the heat kernel on
$\Lambda$, and we strongly believe that one would have the analogue of Theorem \ref{thm. discrete CLT}.  Again, these
considerations will be undertaken elsewhere.

\vspace{5mm}
\noindent
Gautam Chinta \\
 Department of Mathematics \\
 The City College of New York \\
 Convent Avenue at 138th Street \\
 New York, NY 10031 U.S.A. \\
 e-mail: gchinta@ccny.cuny.edu

\vspace{5mm}
\noindent
Jay Jorgenson \\
 Department of Mathematics \\
 The City College of New York \\
 Convent Avenue at 138th Street \\
 New York, NY 10031 U.S.A. \\
 e-mail: jjorgenson@mindspring.com

\vspace{5mm}
\noindent
Anders Karlsson \\
 Section de math\'{e}matiques\\
 Universit\'{e} de Gen\`{e}ve\\
 Case Postale 64, 1211\\
 Gen\`{e}ve 4, Suisse\\
 e-mail: anders.karlsson@unige.ch \\
 and \\
 Matematiska institutionen \\
 Uppsala universitet \\
Box 256, 751 05 \\
 Uppsala, Sweden \\
 e-mail: anders.karlsson@math.uu.se

\vspace{5mm}
\noindent
Lejla Smajlovi\'{c} \\
 Department of Mathematics and Computer Science\\
 University of Sarajevo\\
 Zmaja od Bosne 35, 71 000 Sarajevo\\
 Bosnia and Herzegovina\\
 e-mail: lejlas@pmf.unsa.ba
\end{document}